\documentclass{ieeetran}
\usepackage{graphicx, graphics}
\usepackage{amsmath,amsfonts, dsfont}
\usepackage{amsthm}
\usepackage{nicematrix}
\usepackage{biblatex}
\usepackage{hyperref}
\usepackage{subcaption}
\usepackage{booktabs}
\usepackage{algpseudocode, algorithmicx, algorithm}

\bibliography{ref}
\DeclareMathOperator{\MLP}{MLP}
\DeclareMathOperator{\GNN}{GNN}
\DeclareMathOperator{\softmax}{softmax}

\newtheorem{theorem}{Theorem}
\newtheorem{lemma}{Lemma}

\begin{document}

\title{A neural drift-plus-penalty algorithm for network power allocation and routing}

\author{Ahmed Rashwan, Keith Briggs, and Chris Budd
\thanks{This work has been submitted to the IEEE for possible publication. Copyright may be transferred without notice, after which this version may no longer be accessible}
\thanks{Ahmed Rashwan and Chris Budd are with the Department of Mathematical Sciences, University of Bath, Claverton Down, Bath BA2 7AY UK (email: ar3009@bath.ac.uk; mascjb@bath.ac.uk).}
\thanks{Keith Briggs is with BT Research, Adastral Park, Martlesham, Ipswich IP5 3RE UK (email: keith.briggs@bt.com).}
\thanks{We gratefully acknowledge the support of the EPSRC Programme Grant EP/V026259/1, 'The mathematics of deep learning'.}
}

\maketitle

\begin{abstract}
    The drift-plus-penalty method is a Lyapunov optimisation technique commonly applied to network routing problems. It reduces the original stochastic planning task to a sequence of greedy optimizations, enabling the design of distributed routing algorithms which stabilize data queues while simultaneously optimizing a specified penalty function. While drift-plus-penalty methods have desirable asymptotic properties, they tend to incur higher network delay than alternative control methods, especially under light network load. In this work, we propose a learned variant of the drift-plus-penalty method that can preserve its theoretical guarantees, while being flexible enough to learn routing strategies directly from a model of the problem. Our approach introduces a novel mechanism for learning routing decisions and employs an optimal transport–based method for link scheduling. Applied to the joint task of transmit-power allocation and data routing, the method achieves consistent improvements over common baselines under a broad set of scenarios.
\end{abstract}

\begin{IEEEkeywords}
Radio resource allocation, Graph Neural Networks, Back-pressure routing, Lyapunov optimisation, Optimal transport
\end{IEEEkeywords}

\section{Introduction}
Efficient routing is a fundamental challenge in communication networks, and has motivated a wide range of algorithmic approaches, including oblivious \cite{oblivious_routing}, randomized \cite{randomized_algs}, and distributed \cite{distr_routing} algorithms. The suitability of any such approach strongly depends on the specific routing scenario and application requirements.

The drift-plus-penalty (DPP) method provides a systematic way for designing routing algorithms that guarantee network stability while optimizing a chosen penalty function \cite{RA_book}. By leveraging Lyapunov optimisation, DPP reduces the problem of planning data routes to a simpler sequence of greedy optimisations. Using this framework, distributed routing schemes have been developed that dynamically select data routes without requiring prior knowledge of the underlying network \cite{BP_paper}, making them particularly effective in ad-hoc \cite{ad-hoc_BP} and device-to-device \cite{d2d_survey} environments. The Back-pressure (BP) algorithm \cite{bp_original_paper} is a classic instance of DPP, it is a simple, distributed method that forwards data based on queue differences between neighbouring nodes in multi-hop networks. While BP enjoys strong asymptotic guarantees, it is well known to incur higher delays than shortest-path methods in practice, especially under light network load \cite{adaptive_BP}. This drawback has motivated extensive work on BP-inspired algorithms that preserve throughput optimality while improving empirical performance, including hand-engineered variants \cite{sojourn-time, spbp} as well as machine learning–based approaches \cite{qlearning_BP, gnn_BP}. Among these, Graph Neural Networks (GNNs) provide a particularly promising direction.

GNNs are a class of machine learning models tailored for graph-structured data \cite{graph_networks, original_gnn_paper}. Their ability to exploit graph symmetries has led to strong results in a wide range of domains \cite{GNN_comb_optim, graphcast}. This makes them particularly well suited to communication problems such as radio resource allocation, where the underlying structure is inherently graphical \cite{wireless_gnn}. Importantly, GNNs exhibit an inductive bias toward modelling short-range dependencies \cite{gnn_bottleneck}, which aligns well with many of the central quantities in communications such as radio interference and channel conditions. These properties make GNNs a natural fit for learning distributed routing algorithms which can adapt to arbitrary network topologies and traffic demands.

In this work, we propose a GNN-based framework for learning DPP routing algorithms. We adopt a centralized-training, distributed-execution paradigm: the models are centrally trained offline using a simulation of the problem, but can later be deployed in real systems where execution can be carried out in a fully distributed manner \cite{QMIX, vdn}. While our approach can broadly apply to the class of stochastic network optimization problems described in \cite{RA_book}, we focus for clarity on the joint problem of transmit-power allocation and data routing. Our method combines two learned components: (i) a backlog function that drives routing decisions, and (ii) a DPP solver that jointly allocates transmit power and schedules links. By bounding the backlog function, we guarantee both network stability and penalty minimization. The solver, in turn, provides an efficient implementation of DPP by integrating GNN methods with Sinkhorn’s algorithm \cite{sinkhorn}, a highly efficient method for solving entropy-regularized optimal transport problems. This combination retains the theoretical guarantees of DPP while remaining flexible enough to handle arbitrary penalty functions and topologies. Importantly, the entire framework is fully differentiable, enabling unsupervised training directly from a model of the underlying network.

We evaluate the proposed method on random geometric networks under three penalty functions: fixed penalty, power consumption, and energy efficiency. Our results show that both the learned backlog function and the DPP solver outperform both classical algorithms as well as existing machine learning–based approaches, consistently achieving lower delay across a broad range of scenarios. \\

\textbf{Notation:}\; We extend vector norms $||\cdot||_p$ and inner-products to $N$-dimensional tensors by applying them to the flattened vector representation. Subscripts are used for tensor indexing, and a colon denotes slicing, with trailing colons omitted (e.g., $X_{i,:,:} = X_i$). We write $\mathbf{1}_C$ for the indicator of a condition $C$, and $\mathds{1} \in \mathbb{R}^n$ for an all-ones vector whose dimension is deduced from context. For a tensor $X$, the superscript $X^+$ denotes its element-wise positive components. We let $\mathbb{R}_+$ denote the non-negative reals and $\mathbb{R}_{++}$ the strictly positive reals. We parametrise our ML model using two sets of learned parameters: $\theta$ and $\phi$. We use subscripts to denote functions which are parametrised by either set of parameters. Similar to actor-critic methods in reinforcement learning \cite{ac_survey}, a separate objective functions is used for training each set of parameters, which will be specified later. 

\subsection{Contributions}
We summarise our main contributions below:
\begin{itemize}
    \item We present a GNN architecture for learning distributed algorithms in routing scenarios. The architecture can learn to generalise across different network topologies and varying numbers of sink nodes.

    \item We establish a sufficient condition for the stability of DPP algorithms and use this result to propose a method for learning throughput-optimal routing algorithms.

    \item We develop an efficient, fully differentiable implementation of DPP that jointly performs power allocation and link scheduling, enabling end-to-end training and practical deployment.  

\end{itemize}

\subsection{Structure}
In Section \ref{sec:set-up}, we introduce the routing problem along with the concepts of network stability and throughput-optimal routing algorithms. In Section \ref{sec:architecture} we present the main ML architecture we will be using for learning distributed routing algorithms. Section \ref{sec:Lyapunov} presents the class of drift-plus-penalty algorithms, some common algorithm variants, along with our proposed learned DPP method. We then develop an efficient, end-end differentiable implementation of DPP algorithms in Section \ref{sec:optim}. Finally, we validate our method using simulations in Section \ref{sec:results}, with conclusions and potential for future work in Section \ref{sec:conc}. Appendices \ref{apndx:stable_proof} and \ref{apndx:OT_proof} provide proofs, while Appendices \ref{appendix:sink_impl} and \ref{appendix:mw_impl} describe implementation details of the scheduling schemes considered.

\section{Problem set-up} \label{sec:set-up}

We consider the joint problem of transmit-power allocation and data routing in a network with the objective of stabilising data queues while minimising a given penalty function.

Let $\mathcal G = (V, L)$ be a network with node indices $V = \{1, \hdots, n\}$ and links $L \subset V^{2}$, and let $C = \{1,\hdots, m \} \subset V$ denote the commodity (sink) nodes. Time is slotted $t \in \mathbb{N}_0$, and at each slot the system maintains a set of queues $Q(t) \in \mathbb{R}^{n\times m}_+$ where each $Q_{i,c}(t)$ contains the data stored at node $i \in V$ which is destined for commodity $c \in C$. We assume $Q(0) = 0$ and set $Q_{i,c}(t) \equiv 0$ if $i=c$ since data is removed from the network upon reaching its destination. At each time step, random external data $A(t) \in \mathbb{R}^{n\times m}_+$ arrive at queues $Q(t)$, and a stochastic network state $S(t) \in \mathcal{S}$ is observed. It is assumed that $\{S(t)\}_t, \{A(t)\}_t$ are i.i.d, with $A(t)$ having finite first and second moments.

Given $(S(t), Q(t))$, a routing algorithm first selects link transmit powers $P(t) \in \Pi \subset \mathbb{R}^{n \times n}$, where $\Pi$ is the set of feasible power allocations and $P_{ij}(t)$ is the power transmitted from nodes $i$ to $j$. The choice of $P$ determines both the link capacities $\kappa (P, S)\in \mathbb{R}^{n \times n}_+$ as well as a network penalty $p(P, S) \in \mathbb{R}$. Data can only be transmitted along the links $L$, so we assume $\kappa_{ij} \equiv P_{ij} \equiv0$ if $(i,j) \notin L$. We further assume that there exists constants $\kappa_\text{max},\, p_\text{min}$ such that $\kappa \leq \kappa_\text{max}$ and $p_\text{min} \leq p$. Following power allocation, the algorithm selects a link schedule $\mu(t) \in \mathbb{R}^{n \times n \times c}$, where $\mu_{i,j,c}(t)$ is the amount of commodity $c$ data transferred from nodes $i$ to $j$, subject to link capacity constraints $\sum _{c} \mu_{i,j,c}(t) \leq \kappa_{ij}(t)$ and queue constraints $\sum_j \mu_{i,j,c}(t) \leq Q_{i,c}(t)$. Queue dynamics for $i \neq c$ then follow
\begin{align} \label{eq:evolve}
    \begin{split}
        Q_{i,c}(t+1) = Q_{i,c}(t) + \delta_{i,c}(t) + A_{i,c}(t),
    \end{split}
\end{align}
where $\delta_{i,c}(t) = \sum _j \mu_{j,i,c}(t)- \sum _j \mu_{i,j,c}(t).$

We assume that data is continuous (fluid) as opposed to discrete (packetised), that $\Pi$ is convex, and that $p$ and $\kappa$ are differentiable w.r.t $P$. These assumptions enable us to learn routing algorithms in an end-end fashion.

A system of queues $\{Q(t)\}_t$ is said to be stable if the average expected queue size is bounded, that is, if $$\limsup_{t \to \infty}\frac{1}{t} \mathbb{E}[||Q(t)||_1] < \infty.$$ Stability depends on the arrival-rate matrix $\lambda = \mathbb{E}[A] \in \mathbb{R}^{n \times m}_+$, and the routing algorithm employed. For any given network, the capacity region $\Lambda$ is defined as the set of all rate matrices $\lambda$ for which there a exists a routing algorithm that stabilises $\{Q(t)\}_t$. A routing algorithm is said to be throughput-optimal if it stabilises all queues with rate matrices in the interior of $\Lambda$. Our goal is to design  throughput-optimal routing algorithms which additionally minimise the network penalty. Formally, this leads to the following optimisation problem:
\begin{align} \label{eq:prob}
    \tag{$\mathbf{P}$}
    \begin{split}
    \min \quad & \sum^\infty_{t=0} \frac{1}{t} \mathbb{E}\left[p(P(t), S(t))\right], \\
    \text{s.t.}\quad & \limsup_{t \to \infty}\frac{1}{t} \mathbb{E}[||Q(t)||_1] < \infty.
    \end{split}
\end{align}
We denote by $p^*$ the optimal value of penalty under a stable algorithm, which is assumed to exist.

\section{A GNN architecture for distributed routing} \label{sec:architecture}
Based on the network model in the previous section, we introduce a GNN architecture for learning distributed algorithms. By combining message passing and pooling operations, the architecture can generalise to arbitrary network topologies $\mathcal G$ and to any number of commodities $C \subset V$. This will serve as the backbone for the learned DPP algorithms developed in the remainder of the paper.

For each node $i \in V$, we introduce a latent state $z^t_i \in \mathbb{R}^k$ that encodes the node’s internal knowledge of the network at time $t$, initialised as $z^0_i = 0$. These latent states are updated over time and can be used directly to compute downstream quantities such as transmit powers $P$ or link schedules $\mu$.

At each time step, we form a queue vector $l^t_{i,c}$ containing features such as the queue length $Q_{i,c}(t)$, a commodity identifier for $c$, and (optionally) the backlog function $U_{i,c}(t)$ defined in Section~\ref{sec:Lyapunov}. To handle a variable number of commodities, we aggregate all queue vectors at node $i$ using a permutation-invariant pooling function:
$$a^t_i = \omega_\theta \left( \{l^t_{i, c}\}_{c\in C} \right), $$
where $\omega_\theta: \mathcal{P}(\mathbb{R}^{|l|}) \to \mathbb{R}^k$ is a learned pooling operator \cite{gnn_pooling} and $\mathcal{P}$ denotes the power-set operator. 

The latent states are then updated via a single GNN message-passing layer followed by a recurrent update:
\begin{align*}
    u^{t\phantom{+ 1}} &= \GNN_\theta(z^t||a^t,\, I,\, E), \\
    z^{t+1} &= \text{GRU}_\theta(z^t,\,u^t)
\end{align*}
where $E \in \mathbb{R}^{e \times |L|}$ are any potential edge features, $||$ denotes concatenation, and GRU is the gated recurrent unit. By using only a single message-passing layer, we ensure that the algorithm can be executed in distributed manner as node states are updated using only their immediate neighbours. 

For pooling, we employ softmax attention pooling \cite{softmax_pool}: $$\omega_\theta(x_1,\, \hdots,\, x_n) = \softmax(\MLP^K_\theta(x)) \cdot \MLP^V_\theta(x).$$ For the GNN layer, we used the GINE model \cite{gine}, which allows us to accommodate the edge features $E$. We found this choice of architecture to work well in practice.

\section{Drift-plus-penalty algorithms}\label{sec:Lyapunov}
In this section, we define the class of DPP algorithms and establish conditions under which they solve \eqref{eq:prob}. We then introduce several common variants of the algorithm before presenting our learned DPP algorithm.

\subsection{DPP optimisation}\label{sec:dpp}
At any time--step $t$, we define the drift-plus-penalty optimisation problem:

\begin{align} \label{eq:dpp}
    \tag{$\mathbf{D}$}
    \begin{split}
    \max_{\mu \geq 0, P \in \Pi}\quad & W(t) \cdot \mu - Vp(P, S(t)) \\
    \text{s.t.} \quad & \sum_c \mu_{i,j,c} \leq \kappa_{i,j}(P,S(t)), \; \forall (i,j) \\
     & \sum_j \mu_{i,j,c} \leq Q_{i,c}(t),\; \forall (i,c)
    \end{split}
\end{align}
where $$W_{i,j,c}(t) = U_{i,c}(t) - U_{j,c}(t),$$ $U(t) = U(Q(t), S(t)) \in \mathbb{R}^{n \times m}$ is a given backlog function, and $V \geq 0$ is a penalty weight balancing the trade-off between minimising penalties and queue size. We define the DPP algorithm $\mathcal{D}(U,V): (Q, S) \mapsto (\mu, P)$ as the routing algorithm that, at each time $t$, computes $(P(t), \mu(t))$ by solving problem \eqref{eq:dpp}.

For appropriate backlog functions $U$, the DPP algorithm $\mathcal{D}(U,V)$ is optimal in the sense that it is throughput-optimal for all $V$, and approaches the minimal penalty $p^*$ as $V \to \infty$. Building on results from \cite{qlearning_BP, BP_paper}, we establish (see the proof in Appendix \ref{apndx:stable_proof}), the following sufficient condition for optimality.

\begin{theorem} \label{thrm:stability}
    If a backlog function $U$ satisfies $$||U(t) - Q(t)||_\infty \leq B$$ for some $B \geq 0$ and all $t$, then $\mathcal{D}(U,V)$ is throughput optimal for all $V \geq 0$, and achieves the optimal expected penalty $p^*$ as $V \to \infty$. These results still hold if the queue constraint $\sum_j \mu_{i,j,c} \leq Q_{i,c}(t)$ is omitted from \eqref{eq:dpp}.
\end{theorem}

Theorem~\ref{thrm:stability}, along with related results \cite{lyapnunov_prob1_stabiliy}, follows from bounding the sum of the queue Lyapunov drift, $||Q(t+1)||_2^2 - ||Q(t)||_2^2$, and the associated penalty — hence the name “drift-plus-penalty.”

DPP algorithms reduce the original stochastic planning problem \eqref{eq:prob}, to a set of deterministic greedy problems \eqref{eq:dpp}, which require no knowledge of the network state or external data distributions. If the backlog constraint is omitted, then $\mathcal{D}(Q, 0)$ is the Back-pressure (BP) routing algorithm \cite{bp_original_paper}. As later discussed in Section \ref{sec:optim}, the BP algorithm can be implemented in a distributed manner without knowledge of the network topology, using a simple closed-form solution of \eqref{eq:dpp}.

While theoretically appealing and simple to implement, the BP algorithm tends to exhibit long network delay when compared to shortest-path based routing algorithms \cite{distance_route, geo_route}, especially when backlogs are small. In the extreme case when the network contains only a single packet, BP routing follows a random walk on $\mathcal{G}$, ignoring the packet’s destination. These shortcomings have motivated the design of alternative backlog functions $U$ which reduce network delay while preserving queue stability.

\subsection{Backlog functions} \label{sec:backlog}
In DPP algorithms, routes from sources to commodities are not predetermined, but chosen dynamically via the backlog function $U$, allowing routes to adapt to changing network conditions. Backlog functions are often designed such that each component $U_{i,c}(t)$ is some measure of the time taken to route data from $i$ to $c$. We first introduce two common backlogs: the shortest-path and queue-biased functions, before describing our proposed learned backlog function.

\subsubsection{Shortest-path backlog}
The shortest-path (SP) backlog \cite{spbp} is defined as $$U^\text{SP}(t) = Q(t) + cD,$$ where $c>0$ and $D \in \mathbb{N}^{n \times m}_0$ is the matrix of shortest-path distances between nodes and commodities, which can be computed in a distributed manner using Dijkstra's algorithm \cite{Dijkstra}. 

The SP backlog bias routing decisions toward minimizing distance to commodities, and its effect is most pronounced when queue sizes are small. For a single-packet network, $\mathcal{D}(U^\text{SP}, 0)$ now follows the shortest-path from source to sink. Since $||D||_\infty \leq n-1$ for a connected network, it follows by Theorem \ref{thrm:stability} that this backlog yields an optimal DPP algorithm.

Although SP significantly improves upon the basic BP method, it does not explicitly account for network congestion: the shortest path between nodes may not be optimal when link capacities are heavily loaded.

\subsubsection{Queue-biased paths}
To better capture the effects of congestion, \cite{enhanced_delay_BP} proposed incorporating queue sizes into path selection. Following this idea, we define the queue-biased shortest-path (QSP) backlog as
$$U^\text{QSP}_{i,c}(t) = \min_{\pi \in \Gamma_{i,c}} \sum\nolimits^{|\pi|}_{j =  1} \mathbb{E}[Q_{\pi_j,c}(t + j - 1)],$$
where $\Gamma_{i,c}$ is the set of all paths $\pi = (i, j_2, j_3, \hdots, c)$ from $i$ to $c$ where $\pi_{j + 1} \in N_j$. Note that $U^\text{QSP}_{i,c} \equiv 0$ if $i=c$.

It can be seen that the QSP backlog satisfies the Bellman equation:
$$U^\text{QSP}_{i,c}(t) = \mathbb{E}\big[ Q_{i,c}(t) +  \min_{j \in N_i} U^\text{QSP}_{j,c}(t+1)\big].$$

In analogy to Q-learning \cite{dql} and using a similar method to \cite{qlearning_BP}, we may use the Bellman equation to estimate the QSP backlog as a neural network (with reference to Section \ref{sec:architecture})
\begin{equation} \label{eq:qsp}
    U^\text{QSP}_{i,c}(t) = \MLP^\text{QSP}_\theta(z^t_i||l^t_{i,c})
\end{equation}
 trained to minimise the temporal difference $$\tau =\frac{1}{nm} \sum_{i,c} \left(U^\text{QSP}_{i,c}(t) - Q_{i,c}(t) - \min_{j \in N_i} U^\text{QSP}_{j,c}(t+1) \right)^2.$$ 

The QSP backlog is not bounded from $Q(t)$ in general and is hence not necessary optimal.

\subsubsection{Neural backlogs} \label{sec:neural_bl}
Instead of designing a specific backlog function, we propose an end-end method for learning $U$ directly from a model of the problem. We do this simply by using a learned backlog $U^\text{N}$ which is computed using \eqref{eq:qsp}. We consider two variants of this backlog: Neural-B, which uses a sigmoid output layer to satisfy the bound $||Q - U^\text{N}||_\infty \leq B$, and Neural, which instead uses a linear output layer and hence does not satisfy Theorem \ref{thrm:stability} in general.

Under the smoothness assumptions of Section~\ref{sec:set-up} and assuming that $\mathcal{D}(U,V)$ is differentiable w.r.t $U$, $U^\text{N}$ can be trained directly using any loss function of $Q(t)$ or $\mu(t)$. In this work, neural backlogs are trained to minimise the queue-size loss $$L(Q(t)) = ||Q(t)||_1,$$ though other metrics such as average link utilisation or packet delay~\cite{first-packet_sojourn} are also possible. To enable this approach, we now present an efficient, end-end differentiable implementation of DPP algorithms.

\section{Implementing DPP algorithms} \label{sec:optim}
We propose an ML-based method for obtaining approximate solutions to \eqref{eq:dpp}, thereby enabling practical implementations of the DPP algorithm $\mathcal{D}(U,V)$ for any given backlog $U$ and penalty weight $V\geq 0$. Our approach decomposes the problem into two stages: power allocation, which is done using an ML-method, and link scheduling, which is computed using optimal transport methods.

\subsection{Transmit power allocation}
The transmit powers $P$ are computed using a GNN model which considers the network topology, queues, state, and backlog function:
$$P = f_\phi(\mathcal{G}, Q, U, S).$$
The network $f_\phi$ follows the architecture described in Section~\ref{sec:architecture}, extended with an output layer that maps node states $z$ to power allocations. Many methods exists for enforcing convex constraints $P \in \Pi$ \cite{approx_guarantee, vector_clip, linsatnet}, with the choice depending on the structure of $\Pi$. Here, we assume node-wise power constraints of the form $$\Pi = \left\{ P\in \mathbb{R}^{n \times n}_+ : \sum\nolimits_j  P_{ij} \leq P_\text{max}, \  \forall i \right\}$$ where $P_\text{max} > 0$ is an upper bound on the transmit power of each node. Such constraints are simple enough that feasibility can be enforced directly through the output layer:
\begin{align*}
    P_{ij} = P_\text{max}\,\softmax \left(\MLP^P_\phi (\bar{z}^t_i)\, ||\, \MLP^\text{slack}_\phi (z^t_i) \right)_j,
\end{align*}
where $\bar{z}^t_i = (z^t_i||z^{t-1}_k)_{k \in N_i}.$ The following section outlines how $\phi$ is trained.

\subsection{Scheduling data using optimal transport} \label{sec:schedule}
Once $P$ has been determined, problem~\eqref{eq:dpp} reduces to an optimization over the scheduling variables $\mu$. This yields a linear program that decomposes row-wise over $\mu_i$:

\begin{align} \label{eq:schedule}
    \tag{$\mathbf{W}$}
    \begin{split}
        \max_{\mu_i \geq 0} \quad & W_{i} \cdot \mu_{i}\\ 
        \text{s.t.} \quad &  \mu_{i}\, \mathds{1} \leq \kappa_{i}(P,S), \\
        &\mu_{i}^T \mathds{1} \leq Q_{i},
    \end{split}
\end{align}
where, for brevity, the explicit time dependence is omitted.

If the queue constraint is omitted, the solution to \eqref{eq:schedule} reduces to the max-weight schedule:

\begin{equation} \label{eq:mw}
    \mu_{i, j, c} = \kappa_{ij}(P,S)\; \mathbf{1}_{\{c = \arg \max_{c'} W_{i, j, c'} \}},
\end{equation}
which is commonly used to implement the BP algorithm.

While max-weight allocation is sufficient for guaranteeing the optimality of 
$\mathcal{D}(U,V)$, it can result in significant underutilization of link capacity when queue lengths are small. To address this, we present a method for obtaining fast, approximate solutions to \eqref{eq:schedule} by using methods from discrete optimal transport. 

Let $s_i = \mathds{1} \cdot \kappa_{i}(P,S)$ and $q_i = \mathds{1} \cdot Q_{i}$. Introducing slack variables $\gamma \in \mathbb{R}^{n}_+,\, \sigma \in \mathbb{R}^{c}_+$, we define the following problem, closely related to \eqref{eq:schedule}:

\begin{align} \label{eq:transport}
    \tag{$\mathbf{T}$}
    \begin{split}
    \max_{\mu_i,\, \gamma,\, \sigma \geq 0}\quad & W_{i}^+ \cdot \mu_{i} \\ 
    \text{s.t.}\quad & \mu_{i}\, \mathds{1} + \gamma = \kappa_{i}(P,S),\\
    & \mathds{1} \cdot \gamma = (s_i - q_i)^+,\\
    &\mu_{i}^T \mathds{1} + \sigma = Q_{i},
    \\ & \mathds{1}\cdot \sigma =(q_i - s_i)^+.
    \end{split}
\end{align}
This can be interpreted as a discrete optimal transport problem with unnormalized source and target distributions $$I_i = Q_i\, ||\, (s_i - q_i)^+ \ \text{ and } \ T_i = \kappa_i \, || \, (q_i - s_i)^+,$$ respectively. The following theorem (proven in Appendix \ref{apndx:OT_proof}) describes the relationship between these two optimisation problems.

\begin{theorem} \label{thrm:OT}
    Let $(\mu^\textbf{T}_i, \gamma^\textbf{T}, \sigma^\textbf{T})$ be a solution of \eqref{eq:transport}, then 

    \begin{equation} \label{OT_to_MW}
        \mu^\textbf{W}_i = \left(\mu^\textbf{T}_{i,j,c}\; \mathbf{1}_{\{ W_{i,j,c} > 0 \}} \right)_{j,c}
    \end{equation}
    
    is a solution of \eqref{eq:schedule}.
\end{theorem}

While the OT problem \eqref{eq:transport} is no easier to solve than the original problem \eqref{eq:schedule}, it is possible to use Sinkhorn's algorithm \cite{sinkhorn} to solve the following entropy-regularised variant of \eqref{eq:transport}:

\begin{align} \label{eq:entr}
    \tag{$\mathbf{E}$}
    \begin{split}
    \max_{\mu_i,\, \gamma,\, \sigma \geq 0}\quad & W_{i}^+ \cdot \mu_{i} - \\& \frac{1}{\eta} (\mu_{i} \cdot \log{\mu_{i}} + \gamma \cdot \log\gamma + \sigma \cdot \log \sigma )\\ 
    \text{s.t.}\quad & \mu_{i}\, \mathds{1} + \gamma = \kappa_{i}(P,S),\\
    & \mathds{1} \cdot \gamma = (s_i - q_i)^+,\\
    &\mu_{i}^T \mathds{1} + \sigma = Q_{i},
    \\ & \mathds{1} \cdot \sigma =(q_i - s_i)^+.
    \end{split}
\end{align}
where $\eta > 0$ controls the degree of regularisation.

Sinkhorn's algorithm is a simple scheme which solves problem \eqref{eq:entr} by iteratively normalising the rows and columns of $$M_i = \exp(\eta \bar{W}_i),$$ where $$\bar{W}_i = [(W_i\, ||\, 0)^T\, ||\, 0 ]^T,$$ so that the row and column sums match $T_i$ and $I_i$, respectively. The algorithm is known to be significantly faster than solving the original linear program \eqref{eq:transport}, and can be efficiently parallelised on GPUs, enabling the simultaneous solution of many problem instances~\cite{sinkhorn_permutation}. We therefore propose to compute the link schedule $\mu$ by applying Sinkhorn's algorithm at each node $i \in V$ to obtain the local schedule $\mu'_i$. Then, following \eqref{OT_to_MW}, we set $\mu = (\mu'_{i,j,c}\; \mathbf{1}_{\{ W_{i,j,c} > 0 \}})_{i,j,c}.$

Algorithm~\ref{alg:dpp} provides pseudo-code for our DPP algorithm implementation. Note that the inputs to Sinkhorn can be slightly simplified since either $I_{i,c + 1}$ or $T_{i,n + 1}$ are zero. The algorithm is differentiable w.r.t both $\mu$ and $P$, allowing us to train the parameters $(\phi, \theta)$ as usual using automatic differentiation software. In particular, $\phi$ is trained to maximise the objective of \eqref{eq:entr}.

Beyond its computational efficiency, Sinkhorn offers two further advantages as a surrogate for solving \eqref{eq:schedule}. First, observe that solutions to \eqref{eq:schedule} are piecewise constant as a function of the weights $W$. Consequently, $\nabla_\theta L \equiv 0$ almost everywhere, preventing the training of $\theta$. The entropy regularisation in \eqref{eq:entr} removes this degeneracy and restores useful gradients. Second, we find empirically that entropy regularisation directly improves the performance of some routing algorithms (Section~\ref{sec:results}). A plausible explanation is that regularisation encourages traffic to be distributed more evenly across available paths from $i$ to $c$, rather than concentrating on a single path, thereby alleviating congestion.

\begin{algorithm}
    \caption{Proposed forward pass of the DPP algorithm $\mathcal{D}(U, V)$. The for-loop over $V$ is shown for only for simplifying the exposition. In practice, the loop is replaced with vectorised GPU operations (see Appendix \ref{appendix:sink_impl}).}
    \label{alg:dpp}
    \begin{algorithmic}
        \item \textbf{Input}: $Q,\, S, \, \eta,\, \mathcal{G}, z^t_U, z^t_P$
        \item \textbf{Output}: $P,\, \mu, , z^{t+1}_U, z^{t+1}_P$
        \State $U,\, z^{t+1}_U \gets U^N_\theta(\mathcal{G}, Q, S, z^t_U)$
        \State $P,\, z^{t+1}_P \gets f_\phi(\mathcal{G}, Q, U, S, z^{t}_P)$
        \State $\kappa \gets \kappa(P,S)$
        \State $W_{i,j,c} \gets U_{i,c} - U_{j,c}, \ \forall (j, c)$
        \For{$i \in V$}
        \State $d = \mathds1 \cdot \kappa_{i} - \mathds1 \cdot Q_{i}$
        \If{$d > 0$}
            \State $\bar{W}_i \gets [W_i\ ||\ 0]$
            \State $I_i \gets [Q^T_i\ ||\ d]^T$
            \State $T_i \gets \kappa_i$
        \Else
            \State $\bar{W}_i \gets [W_i^T\ ||\ 0]^T $
            \State $I_i \gets Q_i$
            \State $T_i \gets [\kappa^T_i\ ||\ -d]^T$
        \EndIf
        \State $\mu_i \gets \text{Sinkhorn}(\eta \bar{W}_i, I_i, T_i)_{:n,\, :c}$
        \State $\mu_{i,j,c} \gets \mu_{i,j,c}\, \mathbf{1}_{\{W_{i,j,c} > 0 \}}, \ \forall (j,c)$ 
        \EndFor
    \end{algorithmic}
\end{algorithm}

\begin{figure*}[h]
    \centering\hspace{-8mm}
    \begin{subfigure}[t]{0.265\linewidth}
    \includegraphics[width=\linewidth]{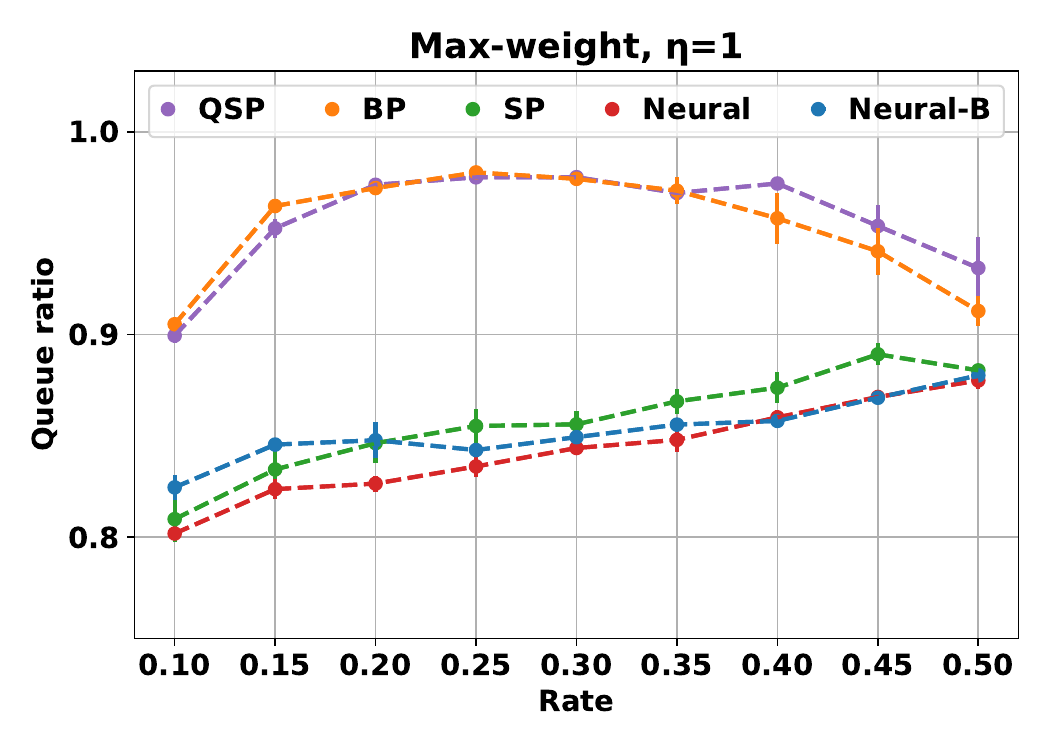}
    \end{subfigure} \hspace{-2.6mm}%
    \begin{subfigure}[t]{0.265\linewidth}
    \includegraphics[width=\linewidth]{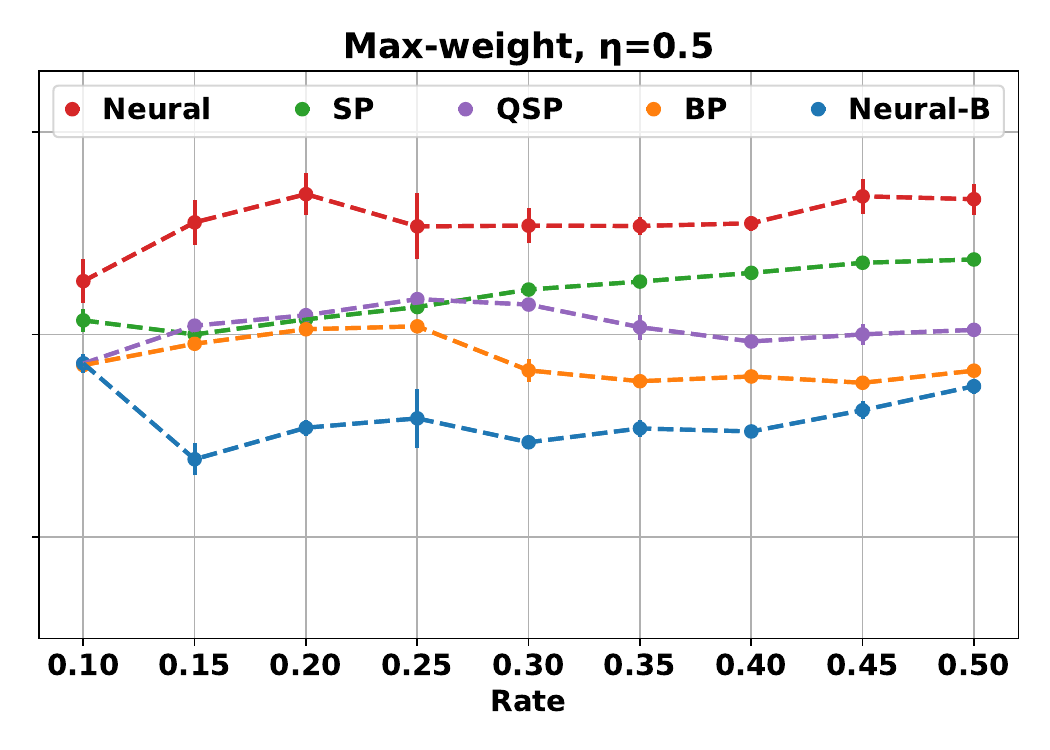}
    \end{subfigure} \hspace{-2.6mm}%
    \begin{subfigure}[t]{0.265\linewidth}
    \includegraphics[width=\linewidth]{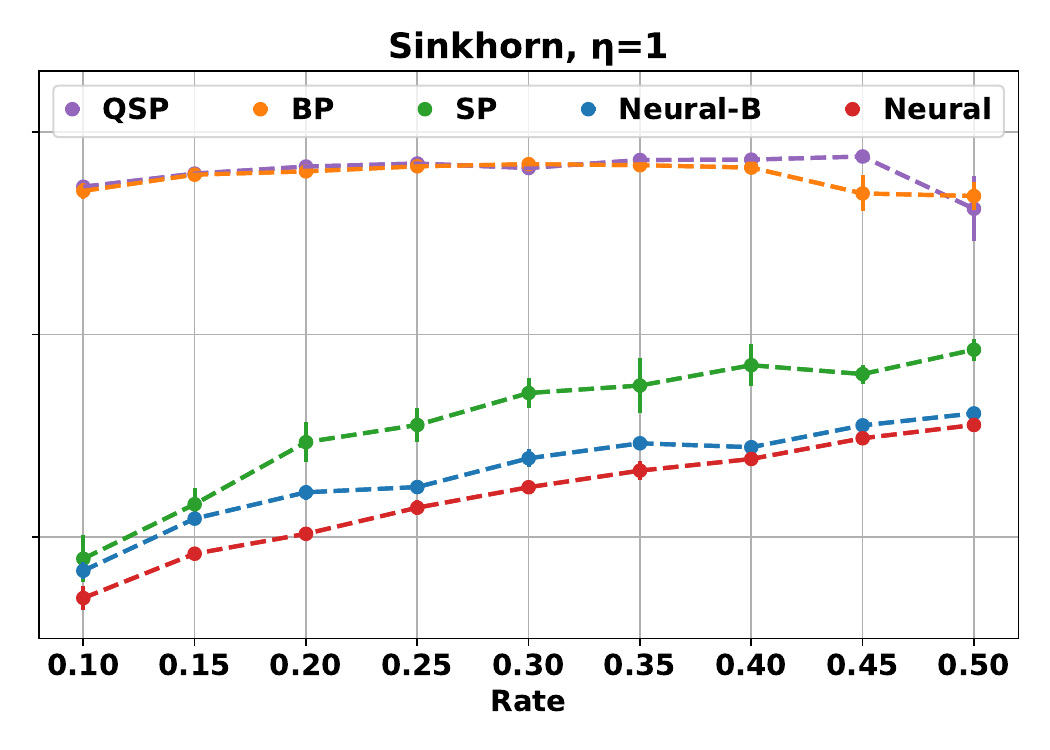}
    \end{subfigure} \hspace{-2.6mm}%
    \begin{subfigure}[t]{0.265\linewidth}
    \includegraphics[width=\linewidth]{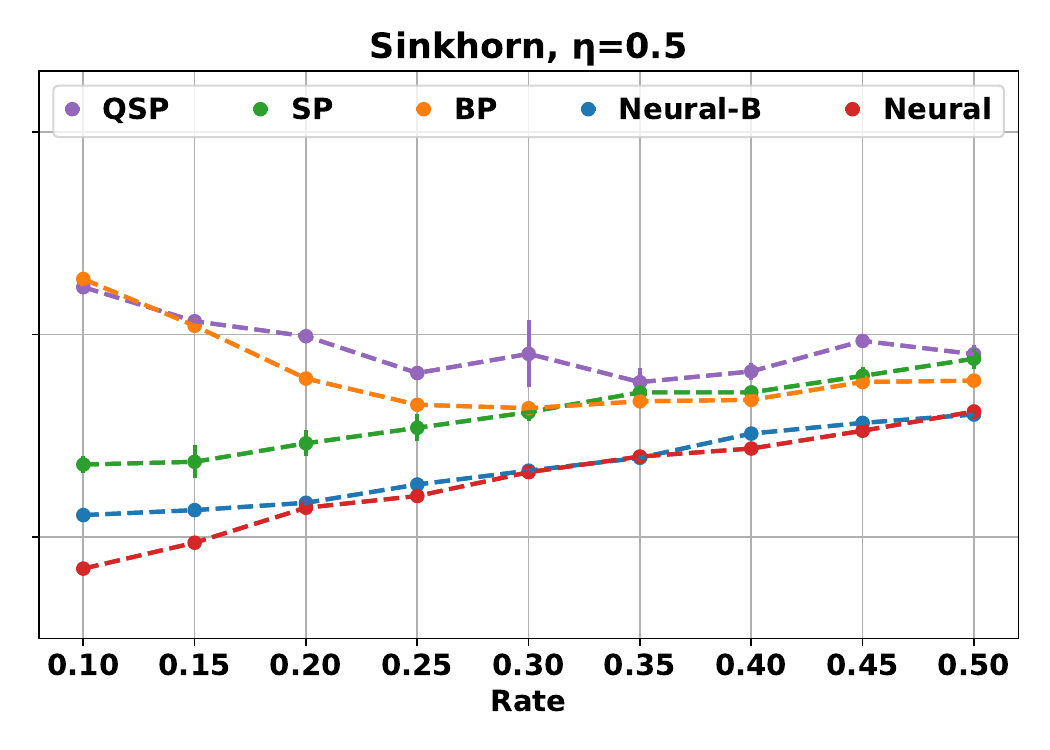}
    \end{subfigure}
    \caption{Plots of queue ratio as a function of data rate $\lambda$ for the Sinkhorn and Max-weight scheduling schemes with regularisation values $\eta \in \{0.5, 1\}$. Error bars show the standard error for each reading. The y-axis is shared across all plots.}
    \label{fig:none}
\end{figure*}

\section{Experiments}\label{sec:results}
We evaluate our approach by comparing neural backlog functions against the alternative backlog functions presented in Section~\ref{sec:backlog}, and by comparing the Sinkhorn-based scheduling method with the a variant of Max-weight (see Appendix \ref{appendix:mw_impl}).We also study the impact of the regularization parameter $\eta$, which influences not only Sinkhorn but also the Max-weight schedule via the objective used for $\phi$. For the Neural-B backlog, we fix $A = 10$.

Performance is reported using two metrics: the average penalty $\mathbb{E}[p]$, and the queue ratio:
$$ r = \mathbb{E}\left[\frac{||Q(t_\text{max})||_1}{\sum^{t_\text{max}-1}_{i=0} ||A(i)||_1} \right] = \frac{\mathbb{E}_{i \neq c} [ Q_{i,c} (t_\text{max})]}{\lambda_0 \, t_\text{max}} \in [0, 1],$$
which measures the total proportion of data still queued in the network after $t_\text{max}$ time-steps. By Little's law, smaller queue ratios indicate lower network delay.

We considered three penalty functions:
\begin{itemize}
    \item No penalty: $p_\text{none}(t) \equiv 0$
    \item Power consumption: $p_\text{cons}(t) = ||P(t)||_1$
    \item Energy efficiency:
    $$p_\text{eff}(t) = -\sum_{i, j} \frac{\kappa_{ij}(t)}{P_{ij}(t) + P_0}\,,$$
    where $P_0 > 0$ is the static power consumption. 
\end{itemize}   

\subsection{Experimental set-up} \label{ssec:generation}

To train and evaluate our methods, we generated random geometric networks $\mathcal G$ with a uniformly sampled number of nodes in the range $[20, 50]$, and an average of $20\%$ of these nodes being commodities. Networks were constructed by uniformly sampling a set of node positions $x_i \in [0,1]^2$, and connecting any two nodes $i,j \in V$ if $ 0<||x_i - x_j||_2 \leq d$, for some threshold $d>0$. Since connectivity is not guaranteed under this construction, some networks may contain disconnected components, in which case certain data may remain isolated from its designated commodity. While inconvenient, this does not affect the validity of our comparisons.

Given transmit powers $P \in \Pi$, the link capacity is computed as 
\begin{align*}
    \kappa_{ij}&(P) \\ = &\log_2\left(1 + \frac{h_{ij}P_{ij}}{\sum_{k\in N_j} h_{kh} \sum_{l \in N_k} P_{kl} - h_{ij}P_{ij} + N_0}\right),
\end{align*}
where $h_{ij} = (1 + ||x_i - x_j ||_2)^{-3}$ are the link channel states, and $N_0>0$ is background noise.

External data $A_{i,c}(t)$ is Poisson distributed with a common $\lambda_0 \in \mathbb{R_+}$, and we assume that the network state $S(t)$ is fixed for all $t$. Each network instance was simulated for a total of 100 time-steps, with the recurrent model described in Section \ref{sec:architecture} trained on blocks of 10 steps. We ran five seeds for each data point reported.

To ensure fair comparisons, all algorithms were implemented in a distributed manner: each node $i \in V$ makes power-allocation and scheduling decisions independently, exchanging at most one message per neighbour per time slot. The ML methods employ the architecture described in Section~\ref{sec:architecture}, which respects these assumptions. For the SP backlog baseline, this was enforced by updating the distance matrix $D$ at each time step via a single iteration of Dijkstra’s algorithm.

\subsection{No penalty results} \label{ssec:allocation}
For $p_\text{none}$, only the queue ratio is relevant. Figure~\ref{fig:none} and Table~\ref{tab:total_perf} confirm show that both Sinkhorn scheduling and neural backlog methods yield meaningful performance gains. Increasing entropy-regularisation (reducing $\eta$) significantly improves the weakest backlogs (BP and QSP). The Neural backlog achieves the best peak performance, closely followed by Neural-B -- which is more robust across $\eta$ values. The benefits of Sinkhorn scheduling are most evident at lower data rates, since Max-weight becomes optimal as $\lambda_0 \to \infty$ when the backlog constraint can be dropped.

\begin{table}[h]
    \centering
    \resizebox{\linewidth}{!}{
    \begin{tabular}{l|c|c|c|c|c}
        \toprule
         
        & BP & SP & QSP & Neural & Neural-B  \\ \midrule
        Max-weight, $\eta = 1$ & 0.953 & 0.857 & 0.957 & 0.843 & 0.852\\
        Max-weight, $\eta = 0.5$ & 0.887 & 0.920 & 0.904 & 0.956 & 0.859\\
        Sinkhorn, $\eta = 1$ & 0.978 & 0.857 & 0.980 & \textbf{0.820} & 0.832 \\
        Sinkhorn, $\eta = 0.5$& 0.880 & 0.861 & 0.894 & 0.827 & 0.834 \\
        \bottomrule
    \end{tabular}
    }
    \caption{Average queue ratio for all combinations of allocation methods (rows) and backlog functions (columns). Ratios are averaged across all tested data rates. All measurement errors are on the order of $10^{-3}$.}
    \label{tab:total_perf}
\end{table}
\subsection{Power consumption and energy efficiency} \label{ssec:backlog_exp}
For the power consumption and energy efficiency  penalties, we recorded both the average penalties and queue ratios across different values of the penalty weight $V \geq 0$. For simplicity, we focused on Sinkhorn scheduling with $\eta = 0.5$, which gave the best average performance, and varied only the backlog functions.

Clear-cut comparisons are challenging, as each method achieves a different balance between the two metrics at the same $V$, even after normalising the weights $W$. For power consumption, neural backlogs consistently maintained lower queue ratios, while Neural-B achieved penalties comparable to the other backlog functions.

\begin{figure}[h]
    \centering\hspace{-5mm}
    \begin{subfigure}[t]{0.52\linewidth}
        \includegraphics[width=\linewidth]{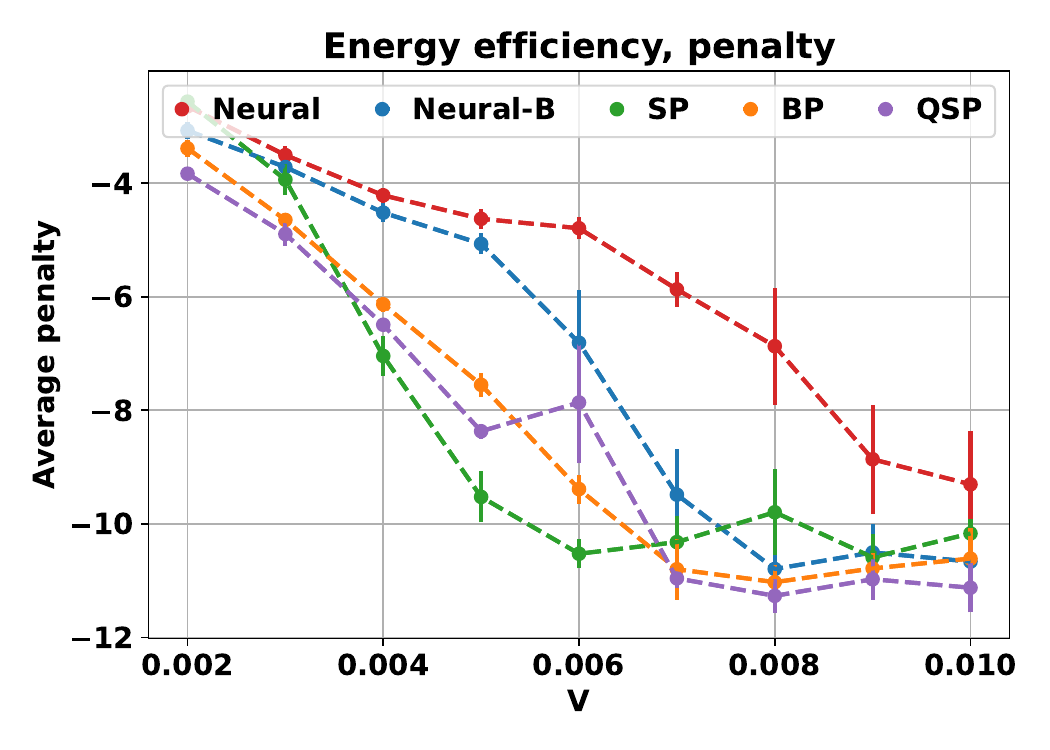}
    \end{subfigure}\hspace{-1mm}%
    \begin{subfigure}[t]{0.52\linewidth}
        \includegraphics[width=\linewidth]{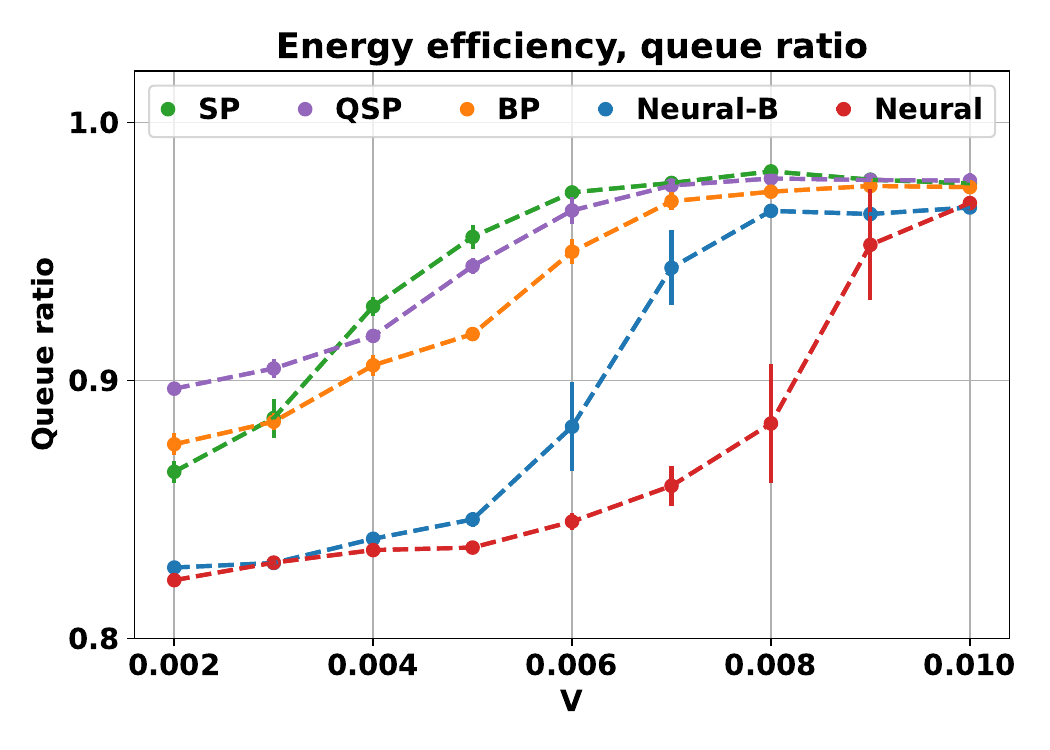}
    \end{subfigure}\vspace{-5mm}
    \begin{subfigure}[t]{0.52\linewidth} \hspace{-4mm}
        \includegraphics[width=\linewidth]{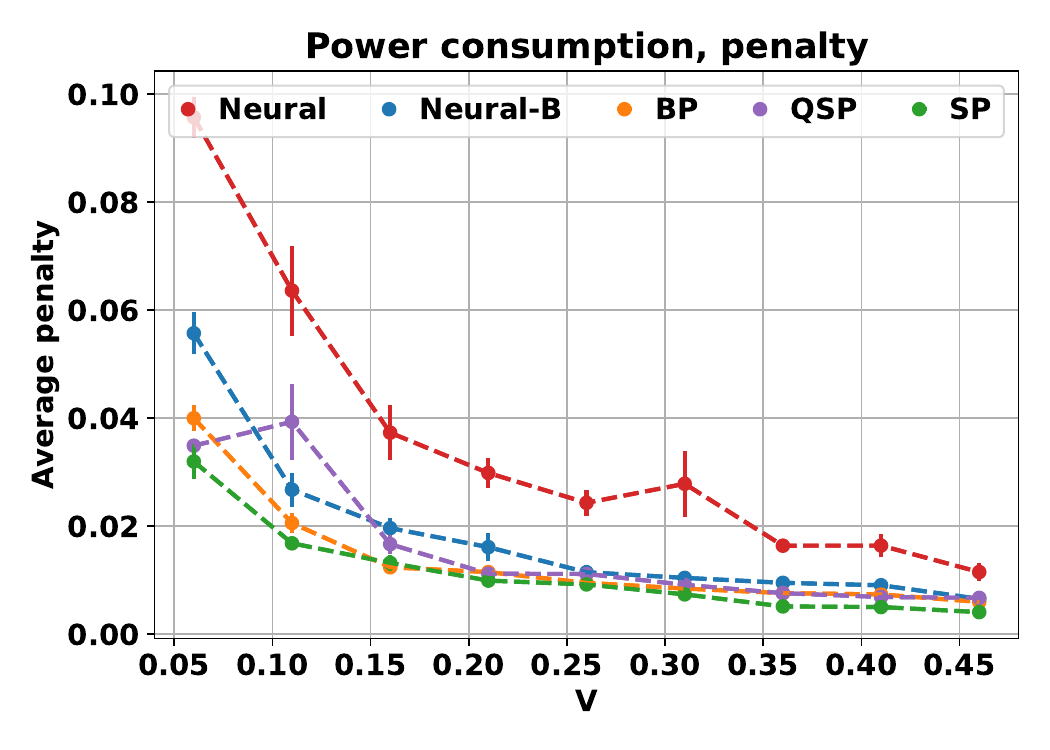}
    \end{subfigure}%
    \begin{subfigure}[t]{0.52\linewidth}
        \hspace{-5mm}
        \includegraphics[width=\linewidth]{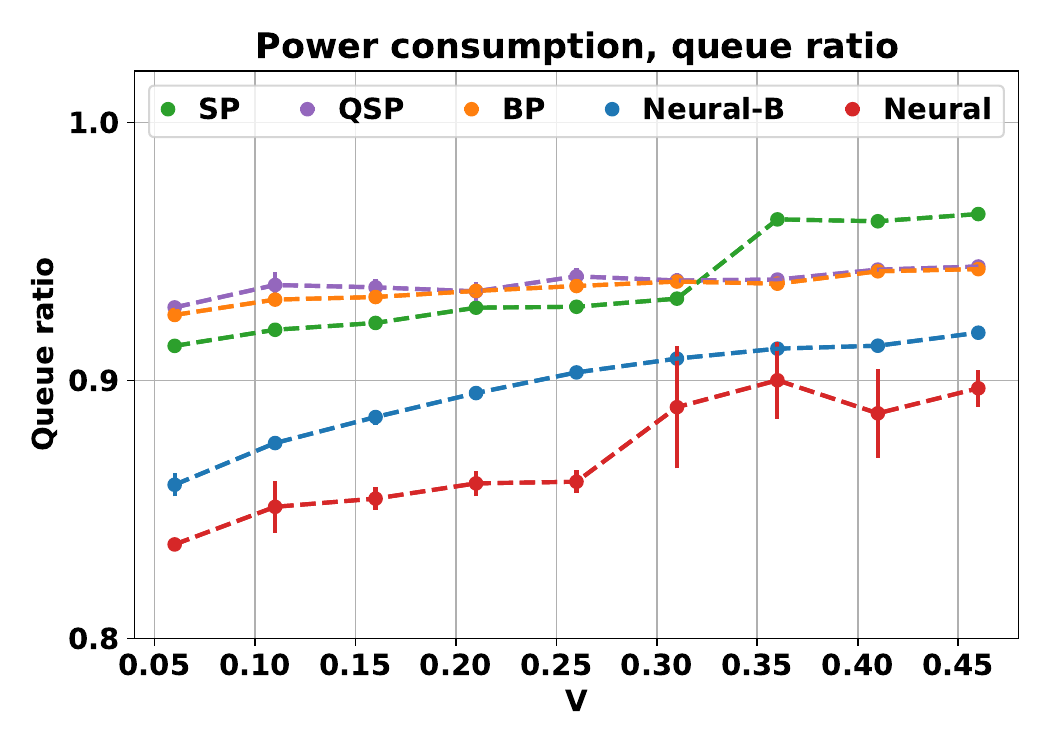}
    \end{subfigure}
    \caption{Average penalty and queue ratios for the power consumption and energy efficiency penalties. All algorithms use the Sinkhorn scheduling method with $\eta = 0.5$. Error bars show the standard error for each reading. }
    \label{fig:ee_pow}
\end{figure}

\subsection{Out-of-distribution results}
In practice, the external data rate $\lambda$ is often unknown at training time, meaning the rates seen during deployment may differ from those used in training. To test robustness under such distribution shifts, we trained no penalty DPP algorithms at a fixed rate $\lambda_0 = 0.25$ and evaluated them at test rates $\lambda_0 \in [0.25, 0.95]$ (Figure~\ref{fig:ood_tvn}).

Although the Neural backlog doesn't ensure stable systems, it consistently achieves lower queue ratios than all other backlogs considered. Taken together with our other results, this supports the consensus that, while stability guarantees are theoretically appealing, they do not always translate into improved empirical performance \cite{empirical_bp, gnn_BP}.

\begin{figure}
    \centering \hspace{-4mm}
    \begin{subfigure}[t]{0.52\linewidth}
        \includegraphics[width=\linewidth]{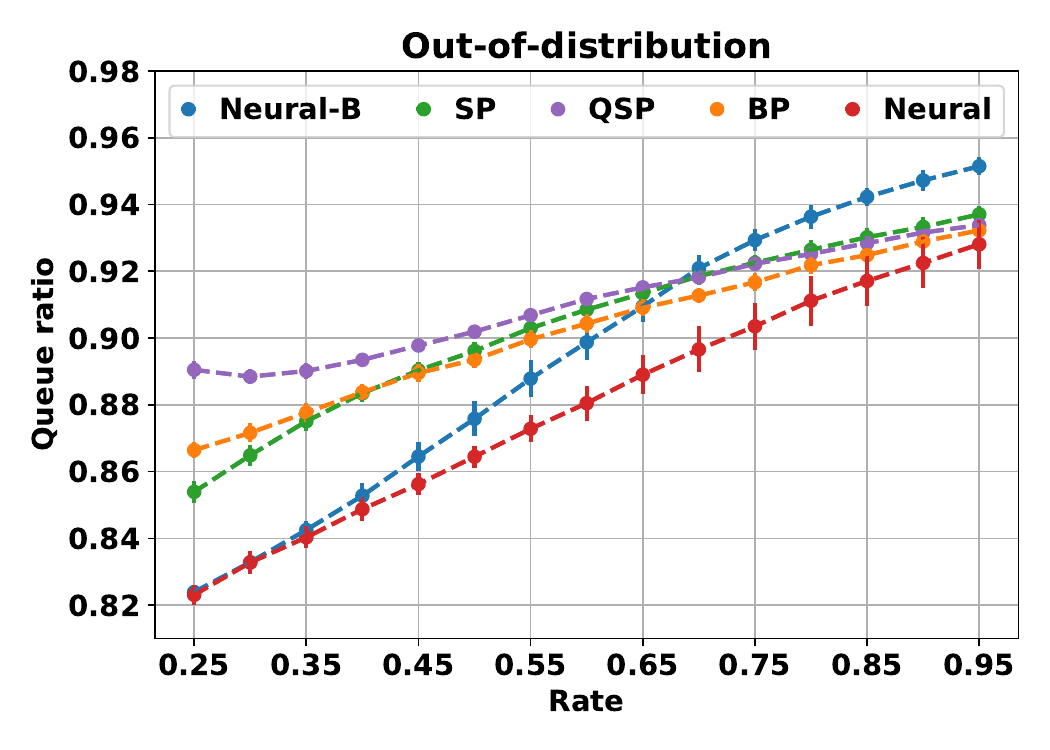}
    \end{subfigure}\hspace{-1mm}%
    \begin{subfigure}[t]{0.52\linewidth}
        \includegraphics[width=\linewidth]{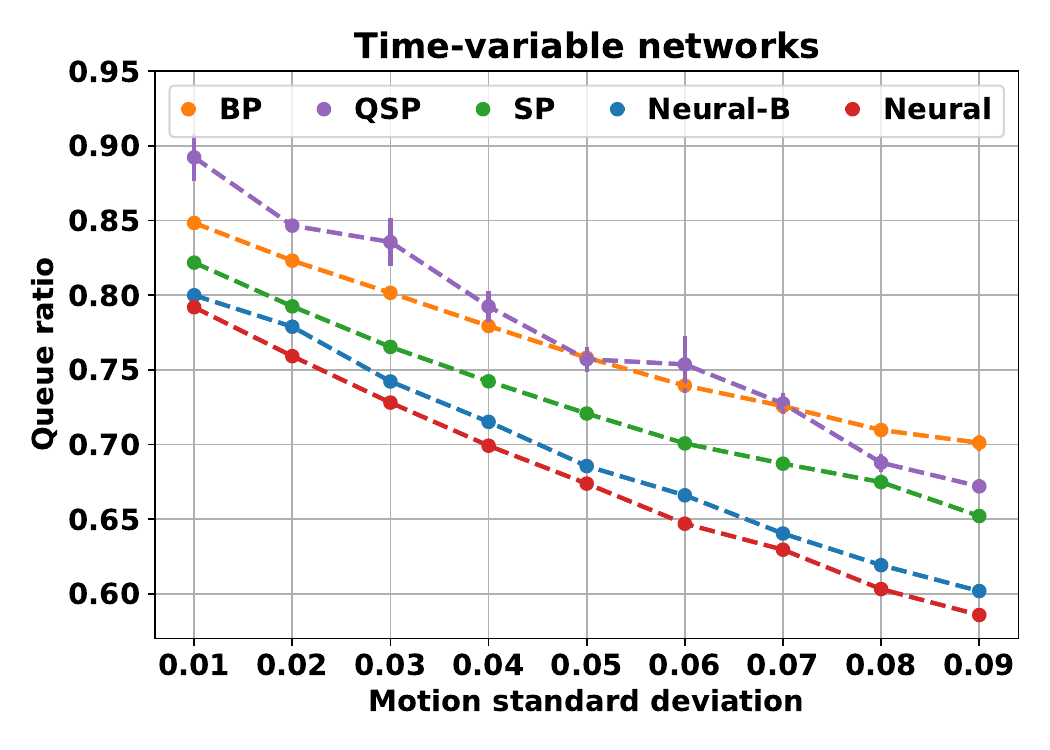}
    \end{subfigure}
    \caption{Results for out-of-distribution and time-variable network settings with no penlty function. All algorithms use the Sinkhorn scheduling method with $\eta = 0.5$. Error bars show the standard error for each reading. }
    \label{fig:ood_tvn}
\end{figure}

\subsection{Time-variable networks}
Thus far, we have assumed a static network topology.  However, if the radio transceivers attached to each node are mobile, then both the channel quality $h$ and the network links $L$ would vary over time. It can be seen that DPP algorithms adapt naturally to this setting without requiring any modifications. To model mobility, we perturbed the positions nodes in a geometric graph according to $$x_{ij}(t+1) = \min\{(x_{ij}(t) + \nu)^+, 1\},$$ where $\nu$ is normally distributed with zero mean. The network links and channel states are then updated based on these perturbed positions.

Figure~\ref{fig:ood_tvn} reports the effect of the standard deviation of $\nu$ on queue ratio across the different backlog methods. While mobility improves performance for all backlogs, neural backlogs benefit the most, with the performance gap relative to the other methods widening as mobility increases.

\section{Conclusion} \label{sec:conc}
We have presented a learned variant of the drift-plus-penalty algorithm that jointly optimizes link transmit power and multi-hop data routing. Our approach integrates two key innovations: a neural backlog function, and an optimal transport–based method for link scheduling. Experiments in a diverse set of network scenarios demonstrated that our method consistently outperforms both classical and existing ML-based alternatives.

\printbibliography
\newpage
\appendix

\subsection{Proof of Theorem \ref{thrm:stability}} \label{apndx:stable_proof}
We define an $S-$only routing algorithm to be any algorithm which determines link power $P(t)$ and transmit data $\mu(t)$ independently of $Q(t)$. If such an algorithm violates the queue inequality, we may assume that excess transmissions are redundant. Recall that $$ \delta_{i,c}(t) = \sum_j\mu_{j,i,c}(t) - \sum_j\mu_{i, j, c}(t).$$ When $\lambda \in \text{int}(\Lambda)$, we assume that there exists an optimal $S-$only algorithm which achieves the minimum penalty $p^*$ and satisfies $\mathbb{E}[\delta_{i,c} + \lambda_{i,c}] < 0$ for $i \neq c$. It can be shown that this assumption is not needed in the no penalty case \cite{RA_book}.

First, we state a fundamental result of Lyapunov optimisation; see \cite{lyapnunov_prob1_stabiliy} for proof.

\begin{lemma} \label{lemma:drift}
    let $\Delta(t) = \frac{1}{2}(||Q(t+1)||_2^2 - ||Q(t)||_2^2)$ denote the Lyapunov drift for a system $\{Q(t)\}_t$ of queues. The queue system is stable if there exists constants $B \in \mathbb{R}$ and $\theta \in \mathbb{R}^{n \times m}_{++}$ such that 
    $$
    \mathbb{E}[\Delta(t) | Q(t)] \leq B - \theta \cdot Q(t)
    $$
    for all $t$.
\end{lemma}

\begin{proof}[Proof of Theorem~\ref{thrm:stability}]
    Define $\delta(t) \in \mathbb{R}^{n \times m}$ as the matrix with entries $\delta_{i,c}(t)$ and let $\bar{\delta}(t) = \delta(t) + \lambda$. Queues evolve according to \eqref{eq:evolve} if $i\neq c$ and are identically zero for $i=c$. Combining these relations we obtain
    $$ Q(t+1) \leq Q(t)+ \delta(t) + A(t).$$
    
    Using this, we compute the following bound on the expected drift:
    \begin{align*}
        &\mathbb{E}[\Delta(t) \mid Q(t)] \leq E+ Q(t) \cdot \mathbb{E}[\bar{\delta}(t) \mid Q(t)] \\
        &=E + \mathbb{E}[(Q(t) - U(t)) \cdot \bar{\delta}(t) +  U(t) \cdot \bar{\delta}(t) \mid Q(t)] \\
        &\leq E+ cB+ \mathbb{E}[U(t) \cdot \bar{\delta}(t) \mid Q(t)],
    \end{align*}

    where the first inequality is standard \cite{BP_paper}, with the constant $E$ depending on the second moment of external arrivals and the maximum link capacity $\kappa_\text{max}$, and $$c = 2\kappa_\text{max}|L| + ||\lambda||_1 \geq ||\bar{\delta}||_1.$$

    We may then compute the following bound on the expected drift-plus-penalty
    
    \begin{align*}
        &\mathbb{E}[\Delta(t) + Vp(t) \mid Q(t)] \\
        &\leq E+ cB+ \mathbb{E}[U(t) \cdot \bar{\delta}(t) + Vp(t)\mid Q(t)] \\
        &=E+ cB- \mathbb{E}[W(t) \cdot \mu(t) - Vp(t) - U(t)\cdot \lambda \mid Q(t) ].
    \end{align*}
    It can be seen that problem \eqref{eq:dpp} is equivalent to opportunistically minimising the above upper bound. Hence, at each time step, the DPP algorithm provides a minimizer of this bound among all routing algorithms.
    
    Assume $\lambda \in \text{int}(\Lambda)$ and that $\mathcal{D}(U,V)$ is used to control the queue system $\{Q(t)\}_t$. Let $(\bar{\delta}^\textbf{S}, p^\textbf{S})$ be the differential flows and penalty value due to the control decisions of the optimal $S$-only algorithm, we then have
    \begin{align} \label{eq:bound}
    \begin{split}
        &\mathbb{E}[\Delta(t) + Vp(t) \mid Q(t)] \\  
        &\leq E+ cB+ \mathbb{E}[ U(t) \cdot \bar{\delta}(t) + Vp(t) \mid Q(t) ] \\
        &\leq E+ cB+\mathbb{E}[U(t) \cdot \bar{\delta}^\textbf{S}(t) + Vp^\textbf{S}(t) \mid Q(t)] \\
        &\leq E+ 2cB+ Q(t) \cdot \mathbb{E}[\bar{\delta}^\textbf{S}(t)]+ V\mathbb{E}[p^\textbf{S}(t)] \\
        &\leq E+ 2cB-\theta \cdot Q(t) + Vp^*.
    \end{split}
    \end{align}
    for constants $\theta_{i, c} > 0$. The second inequality still holds if $\mathcal{D}(U, V)$ omits the queue inequality since the $S-$only algorithm also doesn't consider it.
    
    We then bound drift as $$\mathbb{E}[\Delta(t)|Q(t)] \leq E+ 2cB- \theta\cdot Q(t) +V(p^* -  p_\text{min}),$$
    by lemma \ref{lemma:drift}, it then holds that $\mathcal{D}(U,V)$ is throughput-optimal.
    
    To show that $\mathcal{D}(U,V)$ achieves the minimum penalty as $V \to \infty$, we average inequality \eqref{eq:bound} (marginalised over $Q(t)$) for the first $t$ time-steps:
    \begin{align*}
       &E + 2cB+ Vp^* > \frac{1}{t} \sum^{t-1}_{t' = 0} \mathbb{E}[\Delta(t') + Vp(t')] \\
        & = \frac{1}{t} \bigg[ \big(||L(t)||^2_2 - ||L(0)||^2_2\big)/2 + V\sum^{t-1}_{t'=0}\mathbb{E}[p(t')] \bigg] \\
        & \geq \frac{V}{t}\sum^{t-1}_{t'=0}\mathbb{E}[p(t')],
    \end{align*}
    which yields
    $$ \frac{1}{t}\sum^{t-1}_{t'=0}\mathbb{E}[p(t')] < p^* +  \frac{E + 2cB}{V},$$
    hence $\mathbb{E}[p] \to p^*$ as $V \to \infty$.
\end{proof}

We emphasise that $V$ contributes proportionally to the bound on drift and inversely to the bound on penalty, confirming the intuition that $V$ determines the trade-off between reducing queue size and penalty. Moreover,  $B$ increases the bounds on both the drift and penalty bounds, hence larger values of $B$ allow for a more flexible backlog function $U$ at the cost of robustness.

\subsection{Proof of Theorem \ref{thrm:OT}}\label{apndx:OT_proof}

We first state some properties regarding solutions of \eqref{eq:schedule}.

\begin{lemma} \label{lemma:mw_props}
    Let $\mu^\textbf{W}_i$ be a solution of problem \eqref{eq:schedule} and define
    \begin{align*}
        &E_- = {\{ (j, c) \in  V \times C  : W_{ijc} < 0 \}},\\
        &E_+ = {\{ (j, c) \in  V \times C : W_{ijc} > 0 \}}.
    \end{align*}
    Then the following properties hold:
    \begin{enumerate}
        \item $\mu^\textbf{W}_{i,j,c} = 0$ for all $(j, c) \in E_-$
        \item For all $(j,c) \in E_+$ at least one of the following two inequalities hold:
        \begin{align*}
            &\sum_{c'} \mu^\textbf{W}_{i,j,c'} = \kappa_{ij}(P, S), \ \text{or} \ \
            \sum_{j'} \mu^\textbf{W}_{i,j',c} = Q_{i,c}.
        \end{align*}
    \end{enumerate}
\end{lemma}
\begin{proof}
\begin{enumerate}
    \item This follows since, besides the non-negativity constraint, $\mu$ has only upper-bound constraints. Hence, it is always optimal to set $\mu_{i,j,c} = 0$ if $W_{i,j,c} < 0$.

    \item Assume that neither equality holds at some $(j,c) \in E_+$, it is then possible to increment $\mu^\textbf{W}_{ijc}$ by $$\min\left\{\kappa_{ij}(P,S) - \sum_{c'} \mu^\textbf{W}_{i,j,c'},\,  Q_{i,c} - \sum_{j'} \mu^\textbf{W}_{i,j',c} \right\} > 0,$$
    while retaining feasibility. Since the objective is strictly increasing in $\mu^\textbf{W}_{ijc}$, this contradicts the optimality of $\mu^\textbf{W}_i$.
\end{enumerate}
\end{proof}

\begin{proof}[Proof of Theorem~\ref{thrm:OT}]
    Since $\gamma,\; \sigma \geq 0$, it can be seen that $\mu^\textbf{W}_i$ is feasible in problem \eqref{eq:schedule}.
    
    We then proceed by contradiction. Suppose that $\mu^*_i$ is a solution of \eqref{eq:schedule} such that $$W_i \cdot \mu^*_i> W_i \cdot \mu^\textbf{W}_i. $$
    We show that $\mu^*_i$ can be extended to a feasible solution $(\mu'_i, \gamma', \sigma')$ of problem \eqref{eq:transport} such that $$W_i^+ \cdot \mu'_i > W^+_i \cdot \mu^\textbf{T}_i,$$ contradicting the optimality of $\mu^\textbf{T}_i$.

    To construct $(\mu', \gamma', \sigma')$, we bundle all three variables into one matrix $M \in \mathbb{R}^{(n + 1) \times (c + 1)}_{+}$:
    \begin{equation}
        M =
        \begin{bNiceArray}{cw{c}{1cm}c|c}[margin]
            \Block{3-3}<\large> {\mu'_i } & & & \Block{3-1}{\gamma'} \\
            & & &  \\
            & & &  \\
            \hline
            \Block{1-3} {(\sigma')^T} & & & \epsilon
        \end{bNiceArray},
    \end{equation}
    where $\epsilon$ is a dummy variable.
    
    The constraints in \eqref{eq:transport} may be interpreted as targets on the row and column sums of $M$: $$M \mathds1 = T_i, \quad M^T \mathds1 = I_i.$$
    These constraints are feasible since the sum of all row and column targets are equal
    \begin{equation*}
        \sum_j \kappa_{i,j} + (q_i - s_i)^+ = \sum_c Q_{i,c}(t) + (s_i - q_i)^+.
    \end{equation*}
    We set entries of $M$, or equivalently $(\mu'_i, \gamma', \sigma')$, in three steps:
    \begin{enumerate}
        \item  First, we set $\mu'_{i,j,c} = \mu^*_{i,j,c}$ for all $(j,c) \in E_+$. By Lemma \ref{lemma:mw_props}, each element $(j,c) \in E_+$ saturates either the row or column constraint (or both) at $M_{j,c}$.
        \item We set $M_{ij} = 0$ for all elements where row $i$ or column $j$ is saturated.
        \item Let $M'$ be the residual submatrix obtained after removing saturated rows and columns from M. Its row/column sum constraints are corresponding subvectors of $T_i, I_i$. Choose any feasible $M'$ that satisfies these targets (such an $M'$ exists by the feasibility balance above).
    \end{enumerate}
    
    Since $M$ then satisfies the row and column sum constraints, the corresponding values of $(\mu'_i, \gamma', \sigma')$ are feasible in \eqref{eq:transport}. From lemma \ref{lemma:mw_props}, it follows that $W_i \cdot \mu^*_i= W^+_i \cdot \mu^*_i$. By construction, we have
    $$W^+_i \cdot \mu'_i = W^+_i \cdot \mu^*_i  >W^+_i \mu^\textbf{W}_i  = W^+_i \mu^\textbf{T}_i,$$ contradicting the optimality of $\mu^\mathbf{T}_i$.

\end{proof}

\subsection{Sinkhorn implementation} \label{appendix:sink_impl}
Algorithm \ref{alg:dpp} showns that Sinkhorn's algorithm is applied sequentially to each node $i \in V$. As noted in the caption, this approach is inefficient and presented only for clarity. We now describe a more efficient way to multiplex Sinkhorn computations. Since the input matrices $M_i$ generally have different dimensions, they cannot be simply batched along a third axis. Instead, for inputs $M_1, \hdots, M_n$ we construct a block-diagonal matri
\begin{equation}
    M  = 
    \begin{bmatrix}
        M_1 & & \\
        & \ddots & \\
        & & M_n
    \end{bmatrix}
\end{equation}
and concatenate the row and column sum targets $I_i, T_i$ accordingly. This allows the for-loop in Algorithm \ref{alg:dpp} to be replaced by a batched computation of $(\bar{W}, T, I)$ followed by a single call to sparse Sinkhorn on $M$.

We then implement a GPU-accelerated, sparse variant of Sinkhorn's algorithm. Our implementation leverages the scattering operation \cite{GPU_scatter} $$(\textbf{scatter}(V,I))_i = \sum_{j:I_j = i}V_j.$$ which has efficient GPU implementations in most common ML-frameworks \cite{tensorflow, pytorch}. To improve numerical stability, we also use a log-domain implementation of Sinkhorn. Algorithm \ref{alg:sinkhorn} includes psudo-code for our implementation
\begin{algorithm}
    \caption{Sparse Sinkhorn algorithm. Sparse matrices are represented in coordinate list format using value, row, and column arrays distinguished by subscripts.}
    \label{alg:sinkhorn}
    \begin{algorithmic}
        \State $\textbf{Input}:\  M_v, M_r, M_c, I, T$
        \State $\textbf{Output}: \ N_v$ 
        \State $l_r \gets \log(t_r)$
        \State $l_c\gets \log(t_c)$
        \While{not converged} 
            \State $ \text{sum}_c \gets \log(\textbf{scatter}(\exp(M_v), M_c))$
            \State $M_v \gets M_v + (I - \text{sum}_c)[M_c]$
            \State $ \text{sum}_r \gets \log(\textbf{scatter}(\exp(M_v), M_r))$
            \State $M_v \gets M_v + (T - \text{sum}_r)[M_r]$
        \EndWhile
        \State $N_v \gets \exp(M_v)$
    \end{algorithmic}
\end{algorithm}

\subsection{Max-weight implementation} \label{appendix:mw_impl}
The Max-weight schedule \eqref{eq:mw} cannot be used to train neural backlogs, as it yields no useful gradients with respect to the weights $W$. Moreover, it can transmit more data of a certain commodity than what is available in the queue, which produces some ambiguity as to how schedule ties are broken. To address both issues, we adopt a softmax-based scheme that guarantees the queue constraint in \eqref{eq:bound} is always respected:

\begin{equation*}
    \mu_{i,j,c} = \psi_{i,j,c} \mathbf{1}_{ \{c = \arg \max_{c'} W_{i,j,c'} \} },
\end{equation*}
where $$ \psi_{i,j,c} = \min \{\softmax(W_{i,:,c})_{j}\, Q_{i,c},\; \kappa_{i,j}(P,S) \}.$$
While this rule does not necessarily solve \eqref{eq:schedule} optimally, it is guaranteed to remain feasible, unlike \eqref{eq:mw}.

\end{document}